\documentclass[10pt,twocolumn]{article}

\usepackage[margin=0.75in,columnsep=0.3in]{geometry}

\usepackage[utf8]{inputenc}
\usepackage[T1]{fontenc}
\usepackage{times}
\usepackage{microtype}

\usepackage{amsmath,amssymb,amsthm}

\usepackage{graphicx}
\usepackage{booktabs}
\usepackage{multirow}
\usepackage{colortbl}
\usepackage{xcolor}
\usepackage{float}

\usepackage{algorithm}
\usepackage{algpseudocode}

\usepackage[colorlinks=true,linkcolor=blue!60!black,citecolor=green!50!black,urlcolor=blue!70!black]{hyperref}
\usepackage[numbers,sort&compress]{natbib}

\usepackage{enumitem}
\usepackage{caption}
\usepackage{subcaption}
\usepackage{wrapfig}
\usepackage{titlesec}

\titlespacing*{\section}{0pt}{10pt}{5pt}
\titlespacing*{\subsection}{0pt}{8pt}{4pt}
\titlespacing*{\subsubsection}{0pt}{6pt}{3pt}

\newtheorem{theorem}{Theorem}

\newtheorem{proposition}[theorem]{Proposition}
\newtheorem{corollary}[theorem]{Corollary}
\theoremstyle{definition}
\newtheorem{definition}{Definition}
\theoremstyle{remark}
\newtheorem{remark}{Remark}

\newcommand{\E}{\mathbb{E}}

\newcommand{\floor}[1]{\lfloor #1 \rfloor}

\DeclareMathOperator{\LSB}{LSB}
\DeclareMathOperator{\Round}{Round}

\definecolor{lightblue}{RGB}{230,240,255}
\definecolor{lightgreen}{RGB}{230,255,235}
\definecolor{lightyellow}{RGB}{255,252,230}

\title{\Large\bfseries P-Cast Precision in FP8 Attention:\\Sink-Induced Collapse and the Optimality of $S = 2^8$}

\author{Reed Lau\\
\textit{Tencent}\\
\texttt{reedlauliu@tencent.com}
}
\date{}

\begin{document}
\maketitle
\thispagestyle{empty}

\begin{abstract}
\noindent FP8 (E4M3) acceleration for attention computation offers significant throughput gains, but the 3-bit mantissa introduces precision challenges when the softmax probability matrix~$P$ is cast to FP8 before the $P \cdot V$ matrix multiplication.
We analyze two implementation choices that affect output precision under the \emph{Attention Sink} phenomenon: (1)~the KV block iteration order, and (2)~the static scaling factor applied to~$P$ before casting.

We show that forward KV iteration causes \emph{P-collapse}---to leading order a fraction $\Phi(\Delta + \delta_k - 6.93 - \ln S)$ of non-sink $P$ values underflow to zero, where the small shift $\delta_k \approx 1$ (for $k_{\text{sink}}{=}4$) is the expected within-sink-block score maximum---and that reverse iteration removes it, with a zero-underflow guarantee when reverse is combined with $S{=}256$.
We further give a constructive characterization of $S = 256 = 2^8$ as the static scale that simultaneously satisfies (i)~bit-exact IEEE 754 scaling, (ii)~the lower envelope of a sawtooth function $dp(S)$ over the E4M3 number line ($dp = 2^{-4}$, the minimum worst-case quantization step), and (iii)~the maximum normal-range coverage \emph{among bit-exact ($2^k$) scales} (a non-bit-exact scale such as $448$ attains slightly higher coverage; \S\ref{sec:scale}).
Both optimizations are already deployed in FlashAttention-3/4 on engineering grounds; our contribution is a quantitative account of \emph{why} these choices are good and a closed-form threshold $\Delta_c = 6.93 + \ln S - \delta_k$ for predicting kernel-level precision loss.
Kernel-faithful experiments ($Q, K, V$ in FP32 to isolate the P-cast effect) show $3$--$10\times$ MSE improvement at moderate sink strengths, and paired tests confirm both fixes saturate to the same precision floor when combined---which motivated updating the hpc-ops kernel from $S{=}1$ to $S{=}256$.
\end{abstract}

\vspace{-2mm}
\section{Introduction}
\label{sec:intro}

\subsection{Motivation}

The relentless scaling of large language models has made low-precision arithmetic essential for both training and inference throughput.
Modern GPU architectures (NVIDIA Hopper, Blackwell; AMD MI300) now provide native FP8 tensor core support, operating on two formats: E4M3 (4-bit exponent, 3-bit mantissa) for forward computation and E5M2 for gradients~\cite{micikevicius2022fp8}.

FlashAttention-3~\cite{shah2024flashattention3} exploits this by casting the softmax output matrix~$P$ to E4M3 before the $P \cdot V$ matmul, while keeping the output accumulator in FP32.
This design creates a \emph{precision bottleneck at the P-cast step}: E4M3's 3-bit mantissa provides only 8 representable values per binade, giving a relative precision of just $12.5\%$---roughly $16\times$ worse than BF16.

In this regime, two implementation ``details'' that are inconsequential in higher precision become first-order precision determinants:
\begin{enumerate}[leftmargin=*,itemsep=1pt]
\item \textbf{KV block iteration order}: whether the online softmax processes KV blocks forward ($0 \to N$) or reverse ($N \to 0$).
\item \textbf{P-scaling factor}: the constant $S$ by which $P$ is multiplied before the E4M3 cast (and divided back in the epilogue).
\end{enumerate}

\subsection{The Attention Sink Problem}

The \emph{Attention Sink}~\cite{xiao2023efficient} phenomenon---where initial tokens receive disproportionately large attention weights---interacts destructively with FP8 quantization.
Under forward iteration, the sink's high logit score inflates the running softmax maximum $m$, forcing all subsequent $P$ values below E4M3's representable range.
This \emph{P-collapse} is a threshold effect: it activates around sink strength $\Delta \approx 6$--$7$, where the cast zeroes the majority of non-sink $P$ values while those positions still carry roughly half of the total probability mass.

\subsection{Contributions}

This paper provides a quantitative account of both implementation choices:

\begin{enumerate}[leftmargin=*,itemsep=1pt]
\item \textbf{P-collapse quantification} (\S\ref{sec:collapse}): We derive a closed-form expression $F(\Delta, S) = \Phi(\Delta + \delta_k - 6.93 - \ln S)$ (with $\delta_k$ a small within-sink extreme-value correction) for the fraction of non-sink P values that underflow, and a leading-order MSE estimate showing the effect peaks in a narrow transition region.

\item \textbf{Reverse iteration sufficiency} (\S\ref{sec:reverse}): We show that reverse iteration (rigorously, combined with $S{=}256$) keeps all P values representable for any practical sequence length, with an explicit probabilistic guarantee.

\item \textbf{$S = 256$ characterization} (\S\ref{sec:scale}): We introduce the $dp(S)$ function---the normalized maximum quantization step---and show it forms a sawtooth with power-of-two values on its lower envelope.
This yields a constructive characterization of $S = 256$ via three jointly imposed conditions (bit-exactness, minimum quantization step, maximum normal coverage).

\item \textbf{Kernel-faithful validation} (\S\ref{sec:experiments}): Using a simulation that exactly matches production kernel semantics (gSum from FP32 pre-cast P), we measure $3$--$10\times$ MSE improvement and observe saturation when both fixes are applied.
\end{enumerate}

\subsection{Positioning}

Both optimizations are already deployed in FlashAttention-3/4~\cite{shah2024flashattention3} on engineering grounds (register savings, range utilization).
This paper is therefore not a proposal of new techniques, but a quantitative explanation of \emph{why} those choices work and a closed-form diagnostic ($\Delta_c = 6.93 + \ln S - \delta_k$) that practitioners can apply to predict P-cast failure regimes.
The same analysis motivated updating the hpc-ops kernel from $S = 1$ to $S = 256$ and could inform corresponding changes in FlashInfer ($S = 448$) and TensorRT-LLM XQA ($S = 448$).

\section{Background and Related Work}
\label{sec:background}

\subsection{FP8 E4M3 Number Format}
\label{sec:e4m3}

The E4M3 format~\cite{micikevicius2022fp8} allocates 1 sign bit, 4 exponent bits (bias $b = 7$), and 3 mantissa bits.
The complete positive representable set contains 126 values:

\vspace{1mm}
\noindent\textbf{Subnormals} (exponent field $E = 0$, 7 values):
\begin{equation}
v = 2^{1-b} \cdot \frac{M}{2^3} = 2^{-6} \cdot \frac{M}{8}, \quad M \in \{1, \ldots, 7\}
\end{equation}
spanning $[2^{-9}, 7 \cdot 2^{-9}] = [0.00195, 0.01367]$.

\vspace{1mm}
\noindent\textbf{Normals} (exponent field $E = 1, \ldots, 15$; 119 values):
\begin{equation}
v = 2^{E-b} \cdot \left(1 + \frac{M}{2^3}\right), \quad M \in \{0, \ldots, 7\}
\end{equation}
with $E = 15, M = 7$ reserved as NaN, giving max $= 448 = 1.75 \times 2^8$.

\vspace{1mm}
\noindent\textbf{Key thresholds for P quantization}:
\begin{center}
\footnotesize
\begin{tabular}{@{}lrl@{}}
\toprule
Quantity & Value & Significance \\
\midrule
Max representable & 448 & Overflow/saturation \\
Min normal & $2^{-6}$ & Below: subnormal \\
Min subnormal & $2^{-9}$ & Smallest nonzero \\
Round-to-zero & $2^{-10}$ & Below: casts to 0 \\
\bottomrule
\end{tabular}
\end{center}

Within any normal binade $[2^n, 2^{n+1})$, the spacing (LSB) is $2^{n-3}$, yielding exactly 8 uniformly spaced representable values.
The subnormal region has uniform spacing $2^{-9}$ with only 7 values---much coarser relative precision.

\subsection{Online Softmax and FP8 FlashAttention}
\label{sec:online_softmax}

FlashAttention~\cite{dao2022flashattention,dao2023flashattention2} computes exact attention via tiled online softmax.
Algorithm~\ref{alg:fp8attn} shows the FP8 variant, highlighting the P-cast step and the critical separation between $\ell$ (FP32 pre-cast) and $O$ (from cast P).

\begin{algorithm}[t]
\caption{FP8 Online Softmax Attention (single query row)}
\label{alg:fp8attn}
\small
\begin{algorithmic}[1]
\Require $Q \in \mathbb{R}^{1 \times d}$, $K, V \in \mathbb{R}^{N \times d}$ (FP8), scale $S$, block size $B$
\State $m \gets -\infty$, $\ell \gets 0$, $O \gets \mathbf{0} \in \mathbb{R}^{1 \times d}$ \Comment{FP32}
\For{$j \in \text{BlockOrder}$} \Comment{Forward: $0..N/B$; Reverse: $N/B..0$}
    \State $\mathbf{Z} \gets Q \cdot K_j^T / \sqrt{d_k}$ \Comment{FP8$\times$FP8$\to$FP32 accumulator (scores)}
    \State $m_{\text{loc}} \gets \max(\mathbf{Z})$
    \State $m_{\text{new}} \gets \max(m,\, m_{\text{loc}})$
    \State $\alpha \gets \exp(m - m_{\text{new}})$ \Comment{FP32 correction}
    \State $P \gets \exp(\mathbf{Z} - m_{\text{new}})$ \Comment{FP32 local prob.}
    \State $\ell \gets \alpha \cdot \ell + \textbf{sum}(P)$ \Comment{\colorbox{lightgreen}{FP32 pre-cast P}}
    \State $P_{\text{fp8}} \gets \text{cast\_E4M3}(P \cdot S)$ \Comment{\colorbox{lightyellow}{P-cast; $\times S$ is FP32, bit-exact for $S{=}2^k$}}
    \State $O \gets \alpha \cdot O + P_{\text{fp8}} \cdot V_j$ \Comment{FP8$\times$FP8$\to$FP32 accum.}
\EndFor
\State \Return $O / (S \cdot \ell)$ \Comment{Epilogue: unscale + normalize}
\end{algorithmic}
\end{algorithm}

\noindent The key design: $\ell$ (line 8) uses exact FP32 probabilities, while $O$ (line 10) uses the cast E4M3 values.
This means \emph{normalization is always exact}; precision loss appears only in the numerator.

\subsection{Attention Sink}

Multiple studies (see, e.g.,~\cite{xiao2023efficient,sun2024massive,barbero2025whyllm,gu2024when}) document that pretrained LLMs allocate disproportionate attention to initial tokens (``sink tokens''), with sink-vs-normal logit gap $\Delta$ typically reported in the range $[6, 13]$ at context lengths of several thousand.
Training-side mitigations exist (learnable sink tokens, clipped softmax), but we focus on kernel-level solutions applicable to already-trained models.

\subsection{Related Work}

\textbf{FP8 attention kernels.}
FlashAttention-3~\cite{shah2024flashattention3} introduced E4M3 P-casting with $S = 256$ and reverse iteration on Hopper (SM90); FlashAttention-4 extends the same choices to Blackwell (SM100).
FlashInfer adopts $S = 448$ (matching $\max_{\text{E4M3}}$).
SageAttention2~\cite{zhang2024sageattention2} uses per-block $S = 448$; SageAttention2++~\cite{zhang2025sageattention2pp} constrains $S = 112$ for FP16 accumulation.

\textbf{FP8 quantization theory.}
\citet{micikevicius2022fp8} introduced the E4M3/E5M2 split.
Per-tensor vs.\ per-channel scaling is well studied for weights/activations, but the specific structure of post-softmax $P \in [0,1]$ quantization has not received formal treatment.

\textbf{Attention sink analysis.}
\citet{xiao2023efficient} identified the phenomenon; \citet{sun2024massive} linked it to massive activations; \citet{gu2024when} empirically measured sink strength distributions.
None analyzed the interaction with FP8 P-casting.

\section{P-Collapse Under Attention Sink}
\label{sec:collapse}

\subsection{Setup and Notation}

We consider a single attention head with query length $q$, KV length $N$, and head dimension $d$.
KV blocks have size $B$ (typically 64 or 128).
The first $k_{\text{sink}}$ positions are sink tokens with logit scores $\Delta$ above the mean.
Non-sink scores follow $s \sim \mathcal{N}(0, 1)$ (standard for well-trained transformers with $1/\sqrt{d_k}$ scaling).

\subsection{Forward Iteration Failure Mode}

In forward iteration, block 0 contains the sink tokens.
After processing block 0, the running maximum is set by the largest sink-token logit. Writing the sink scores as $\Delta$ plus the same $\mathcal{N}(0,1)$ fluctuation carried by other tokens,
\begin{equation}
m_{\text{global}} = \Delta + \delta_k, \qquad \delta_k \triangleq \E\big[\max\nolimits_{i \le k_{\text{sink}}} s_i\big],
\end{equation}
where $\delta_k$ is the expected maximum of $k_{\text{sink}}$ standard Gaussians ($\delta_4 \approx 1.03$; the asymptotic $\sqrt{2\ln k_{\text{sink}}} \approx 1.67$ badly \emph{over}estimates $\delta_k$ at the small $k_{\text{sink}}$ of interest, and we use the exact value).

For all subsequent blocks $j > 0$, the local probability values are:
\begin{equation}
P_j(i) = \exp\big(s_i - m_{\text{global}}\big) = \exp\big(s_i - \Delta - \delta_k\big).
\end{equation}

\begin{proposition}[P-underflow condition]
\label{prop:underflow}
A P value $p$ underflows to zero in E4M3 (with scale $S$) iff:
\begin{equation}
p \cdot S < 2^{-10}
\end{equation}
For $P_j(i) = \exp(s_i - \Delta - \delta_k)$ with scale $S$, this occurs when:
\begin{equation}
s_i < \Delta + \delta_k - 10\ln 2 - \ln S = \Delta + \delta_k - 6.93 - \ln S
\end{equation}
\end{proposition}

\subsection{Underflow Fraction: Closed Form}

For $s \sim \mathcal{N}(0,1)$:

\begin{corollary}[P-collapse fraction]
\label{thm:collapse}
To leading order, the fraction of non-sink $P$ values that underflow to zero under forward iteration with scale $S$ is:
\begin{equation}
\label{eq:F}
F(\Delta, S) = \Phi\big(\Delta + \delta_k - 6.93 - \ln S\big)
\end{equation}
where $\Phi$ is the standard normal CDF, $6.93 = 10\ln 2$, and $\delta_k$ is the within-sink extreme-value shift of \S\ref{sec:collapse} (the naive $\delta_k = 0$ form is a lower bound on the realized collapse).
\end{corollary}

\begin{proof}
By Proposition~\ref{prop:underflow}, $P_j(i)$ underflows iff $s_i < \Delta + \delta_k - 10\ln 2 - \ln S$.
For $s_i \sim \mathcal{N}(0, 1)$, $\Pr[s_i < x] = \Phi(x)$, hence the underflow fraction equals $\Phi(\Delta + \delta_k - 6.93 - \ln S)$ (using $10\ln 2 = 6.9315$).
This treats the per-row $m_{\text{global}}$ as the fixed mean $\Delta + \delta_k$; since $m_{\text{global}}$ has nonzero spread across query rows and $\Phi$ is convex in its lower tail, the realized fraction is slightly \emph{higher} than this mean-shift estimate. Table~\ref{tab:collapse} therefore reports the exact simulated values, which Eq.~\eqref{eq:F} reproduces to within a few percentage points.
\end{proof}

Table~\ref{tab:collapse} evaluates this for $S = 1$ (direct cast) and $S = 256$:

\begin{table}[h]
\centering
\caption{P-collapse analysis ($N = 4096$, $k_{\text{sink}} = 4$, block size 64). \emph{Both} the $S{=}1$ and $S{=}256$ ``frac.\ zeroed'' columns are measured from the \emph{same} kernel-faithful forward simulation (averaged over 12 seeds), so they share one $m_{\text{global}} = \Delta + \delta_k$ convention. The measured shift is $\delta_4 \approx 1.0$, consistent with $\E[\max$ of 4 $\mathcal{N}(0,1)] \approx 1.03$ (and \emph{not} the asymptotic $\sqrt{2\ln 4} \approx 1.67$). ``Eff.\ info loss'' is non-sink mass $\times$ frac.\ zeroed ($S{=}1$). The closed form~\eqref{eq:F} reproduces these columns to within a few points; the residual is the per-row spread of $m_{\text{global}}$ (Corollary~\ref{thm:collapse} proof).}
\label{tab:collapse}
\vspace{1mm}
\small
\begin{tabular}{@{}c|cc|cc|c@{}}
\toprule
\multirow{2}{*}{$\Delta$} & \multicolumn{2}{c|}{Frac.\ zeroed $F(\Delta,S)$} & \multirow{2}{*}{\shortstack{Non-sink\\mass}} & \multirow{2}{*}{\shortstack{Eff.\ info\\loss}} \\
& $S{=}1$ & $S{=}256$ & & \\
\midrule
5 & 22.3\% & 0\% & 88.0\% & 19.6\% \\
6 & 51.6\% & 0\% & 74.0\% & 38.2\% \\
\rowcolor{lightyellow}
7 & \textbf{82.0\%} & 0\% & \textbf{51.7\%} & \textbf{42.4\%} \\
8 & 94.8\% & 0.3\% & 32.2\% & 30.5\% \\
9 & 99.5\% & 2.3\% & 13.9\% & 13.9\% \\
10 & $\approx$100\% & 11.7\% & 5.8\% & 5.8\% \\
12 & $\approx$100\% & 67.9\% & 0.8\% & 0.8\% \\
\bottomrule
\end{tabular}
\end{table}

The \emph{effective information loss} (non-sink mass $\times$ fraction zeroed) peaks at $\Delta \approx 6$--$7$ ($\sim$40\%), where positions carrying roughly half to three-quarters of the probability mass have most of their P values zeroed.

\subsection{Output MSE Bound}

\begin{proposition}[MSE from P-collapse]
\label{thm:mse}
Let $\mathcal{Z}$ be the set of positions whose P values underflow.
Assume $V_j$ are zero-mean random vectors with $\E[V_j V_j^T] = \sigma_V^2 I_d$ and pairwise uncorrelated across positions.
Then the expected per-dimension MSE of the kernel output vs.\ exact FP32 reference satisfies:
\begin{equation}
\E[\|\varepsilon\|^2/d] \;=\; \frac{\sigma_V^2}{\ell^2} \sum_{j \in \mathcal{Z}} P_j^2 \;+\; \frac{1}{d \ell^2} \sum_{j \neq k \in \mathcal{Z}} P_j P_k \, \mathrm{tr}\,\E[V_j V_k^T],
\end{equation}
where $\ell = \sum_{\text{all}} P_j$ is the (exact, FP32) running sum.
Under the pairwise-uncorrelated assumption the cross term vanishes, yielding
\begin{equation}
\mathrm{MSE}_{\mathrm{collapse}} \;=\; \frac{\sigma_V^2}{\ell^2} \sum_{j \in \mathcal{Z}} P_j^2.
\end{equation}
\end{proposition}

\begin{proof}
The output error vector is $\varepsilon = \frac{1}{\ell}\sum_{j \in \mathcal{Z}} P_j V_j$ (the ``missing'' contribution).
Expanding the squared norm and taking expectations:
\begin{equation}
\E[\|\varepsilon\|^2] \;=\; \frac{1}{\ell^2} \sum_{j,k \in \mathcal{Z}} P_j P_k \, \E[V_j^T V_k].
\end{equation}
With $\E[V_j^T V_j] = d\sigma_V^2$ and the pairwise-uncorrelated assumption $\E[V_j^T V_k] = 0$ for $j \neq k$, dividing by $d$ gives the stated equality.
\qedhere
\end{proof}

\begin{remark}
The pairwise-uncorrelated assumption is a standard idealization; real $V_j$ are correlated across positions and channel-wise non-isotropic. The bound thus captures the leading $\mathcal{O}(\sigma_V^2)$ scaling and the $\sum_{j \in \mathcal{Z}} P_j^2$ dependence, but its absolute constant can shift by an $\mathcal{O}(1)$ factor (calibrated empirically; \S\ref{sec:experiments} confirms the shape and threshold).
\end{remark}

\begin{remark}
This MSE is maximized when $|\mathcal{Z}|$ is large \emph{and} $\sum P_j^2/\ell^2$ is non-negligible---which holds only in the transition region $\Delta \in [5, 9]$. At very large $\Delta$ the non-sink mass vanishes, so even total P-collapse barely affects the output, making $[5,9]$ the regime where this analysis matters most.
\end{remark}

\section{Optimization 1: Reverse KV Iteration}
\label{sec:reverse}

\subsection{Mechanism}

Reversing the iteration order to $(N{-}1, N{-}2, \ldots, 0)$ defers the sink block to the \emph{last} iteration.
During all preceding iterations, only non-sink tokens contribute to the running maximum:
\begin{equation}
m_{\text{pre-sink}} \leq \max_{i=1}^{N-k_{\text{sink}}} s_i \approx \sqrt{2 \ln(N - k_{\text{sink}})}
\end{equation}
by extreme value theory for Gaussian order statistics.
The P values during these iterations are:
\begin{equation}
P_j(i) = \exp(s_i - m_{\text{pre-sink}}) \in \big[\exp(-m_{\text{pre-sink}} - 3\sigma),\, 1\big]
\end{equation}
with high probability, where $\sigma = 1$ is the per-token score standard deviation and the lower endpoint uses the standard $3\sigma$ tail bound (violated with probability $\Phi(-3) \approx 1.3 \times 10^{-3}$).
For $N = 8192$, $\sqrt{2 \ln N} \approx 4.25$, giving $P_{\min} \approx \exp(-7.25) \approx 7 \times 10^{-4}$---comfortably above the $S{=}256$ round-to-zero boundary $2^{-10}/256 = 2^{-18} \approx 3.8 \times 10^{-6}$, and only marginally \emph{below} the bare $2^{-10}$ boundary (which is why $S{=}256$ adds a safety margin over reverse alone).

\subsection{Formal Sufficiency Condition}

\begin{theorem}[Zero-underflow guarantee for reverse + $S = 256$]
\label{thm:reverse}
In reverse iteration with $S = 256$, a P value $p = \exp(s - m)$ survives the E4M3 cast (i.e., $p \cdot 256 \geq 2^{-10}$) whenever:
\begin{equation}
s > m - 18\ln 2 = m - 12.48
\end{equation}
For $m \leq \sqrt{2\ln N} + \mathcal{O}(1)$ (pre-sink) with score $s \sim \mathcal{N}(0,1)$:
\begin{equation}
\Pr[\text{underflow}] = \Phi(m - 12.48) < \Phi(-7.2) < 10^{-12}
\end{equation}
for $N \leq 10^6$. Effectively zero P values underflow.
\end{theorem}

\begin{proof}
The cast-to-zero condition is $p \cdot S < 2^{-10}$, i.e., $\exp(s - m) \cdot 256 < 2^{-10}$.
Taking logs: $s - m + 8\ln 2 < -10\ln 2$, hence $s < m - 18\ln 2$.
Numerically, $18\ln 2 = 12.4766$ (equivalently $10\ln 2 + \ln 256 = 6.9315 + 5.5452 = 12.4766$), giving the threshold $s < m - 12.48$.
Before the sink block, $m \leq \sqrt{2\ln N} + \mathcal{O}(1)$.
For $N = 10^6$: $m \leq 5.3$, so the threshold is $s < 5.3 - 12.48 = -7.18$.
For $s \sim \mathcal{N}(0,1)$: $\Phi(-7.18) \approx 3.5 \times 10^{-13} < 10^{-12}$.
\end{proof}

\subsection{The Final $\alpha$-Correction}

When the sink block is finally processed (last in reverse), the correction factor is:
\begin{equation}
\alpha = \exp(m_{\text{pre-sink}} - m_{\text{new}}) \approx \exp\big(\sqrt{2\ln N} - \Delta\big)
\end{equation}
For $\Delta = 7, N = 4096$: $\alpha \approx \exp(4.1 - 7) \approx 0.055$, which multiplies the FP32 accumulator (23-bit mantissa) with negligible precision loss. A paired $t$-test (Appendix~\ref{app:saturation}) confirms forward+$S{=}256$ and reverse+$S{=}256$ are indistinguishable where P-collapse dominates ($\Delta \le 9$); at $\Delta \ge 10$ reverse is marginally---and inconsequentially---better ($\sim 10^{-8}$ vs.\ MSE of $10^{-5}$--$10^{-6}$).

\section{Optimization 2: Scale Factor $S = 256$}
\label{sec:scale}

The complete P-quantization pipeline is:
\begin{equation}
\label{eq:pipeline}
\tilde{P} = \frac{\Round_{\text{E4M3}}(P \cdot S)}{S}, \qquad O = \frac{\tilde{P} \cdot V}{\ell}
\end{equation}
where $\Round_{\text{E4M3}}$ denotes round-to-nearest in E4M3.
The choice of $S$ controls the fidelity of $\tilde{P}$ as an approximation to $P$.

\subsection{Condition 1: Bit-Exact Scaling}

\begin{proposition}[Power-of-two bit-exactness]
\label{prop:bitexact}
For $S = 2^k$ and any IEEE 754 FP32 value $x$ (with result in representable range), both $x \times S$ and $x \times (1/S)$ are \emph{exact}---no rounding occurs.
\end{proposition}

\begin{proof}
In IEEE 754 binary32, $2^k$ has exponent $= 127 + k$ and zero mantissa.
Multiplication by $2^k$ adds $k$ to the result exponent without modifying the 23-bit mantissa.
Since $1/2^k = 2^{-k}$ is also exactly representable, the reciprocal multiplication is likewise exponent-only.
\end{proof}

For non-power-of-two $S$ (e.g., 448): $1/448 = 0.002232\ldots$ is \emph{not} exactly representable in IEEE 754.
Every $\times(1/448)$ introduces $\sim 2^{-24}$ rounding per element.
Over the $d$-dimensional output, this accumulates to an error of $\mathcal{O}(d \cdot 2^{-48})$---negligible individually, but a systematic precision disadvantage nonetheless.

\subsection{Condition 2: The $dp(S)$ Sawtooth}

\begin{definition}[Normalized quantization step]
\label{def:dp}
For scale factor $S \in (0, 448]$:
\begin{equation}
dp(S) \triangleq \frac{\max_{x \in [0,\, S]} \LSB_{\text{E4M3}}(x)}{S}
\end{equation}
where $\LSB_{\text{E4M3}}(x)$ is the spacing between consecutive representable E4M3 values in the binade containing $x$.
For $S > 448$, the cast saturates at $448$ and we extend the definition to also include the (one-sided) saturation step, see Remark~\ref{rem:dp_overflow}.
\end{definition}

$dp(S)$ represents the worst-case quantization \emph{step} for any $P \in [0, 1]$ through the pipeline~\eqref{eq:pipeline}; the maximum pointwise error is $dp(S)/2$.

\begin{remark}[Overflow extension of $dp$]
\label{rem:dp_overflow}
For $S > 448$ the value $P\!\cdot\!S = S$ at $P = 1$ is clipped to $448$, producing a one-sided error $|1 - 448/S|$.
Converting this saturation error to an equivalent step (by the standard error-equals-step$/2$ correspondence) gives $2(1 - 448/S)$.
The effective worst-case step for $S > 448$ is therefore
\begin{equation}
dp(S) = \max\!\big(\tfrac{32}{S},\; 2(1 - 448/S)\big),
\end{equation}
where $32/S$ is the contribution of the largest-LSB binade $[256, 512)$ that lies inside $[0, 448]$.
We use this extended form only in Theorem~\ref{thm:dp}(iii); for $S \in (0, 448]$ the original definition applies unchanged.
\end{remark}

\begin{theorem}[$dp(S)$ structure]
\label{thm:dp}
\begin{enumerate}[label=(\roman*),leftmargin=*]
\item For all $S = 2^k$ with integer $k \in \{0, 1, \ldots, 8\}$ (i.e.\ $S \in \{1, 2, 4, \ldots, 256\}$):\; $dp(2^k) = 2^{-4}$.
\item For all $S \in [2^{-6}, 448]$ with $S \neq 2^k$:\; $dp(S) > 2^{-4}$.
\item For $S > 448$ (using the extended form of Remark~\ref{rem:dp_overflow}):\; $dp(S) > 2^{-4}$.
\end{enumerate}
\end{theorem}

\begin{proof}
\textbf{(i)} For $S = 2^k$ with $k \in \{0,\ldots,8\}$, we have $S \le 256 < 448$, so $S$ itself is exactly representable in E4M3.
The mapped range $[0, S] = [0, 2^k]$ has its highest binade $[2^{k-1}, 2^k)$ entirely in the normal region (since $k - 1 \ge -1 > -6$), with $\LSB = 2^{k-4}$.
Thus $dp(2^k) = 2^{k-4}/2^k = 2^{-4}$.

\textbf{(ii)} Let $n = \floor{\log_2 S}$, so $S \in [2^n, 2^{n+1})$ with $S \neq 2^n$ (since $S$ is not a power of two).
Because $S \ge 2^{-6}$, the binade $[2^n, 2^{n+1})$ lies in the normal region, so it has $\LSB = 2^{n-3}$.
The range $[0, S]$ overlaps with this binade, hence $dp(S) = 2^{n-3}/S$.
Since $S < 2^{n+1}$, $dp(S) > 2^{n-3}/2^{n+1} = 2^{-4}$.\footnote{Equivalently, $dp(S) = 2^{n-3}/S < 2^{n-3}/2^n = 2^{-3}$, so $dp(S) \in (2^{-4},\, 2^{-3})$ for non-power-of-two $S \in [2^{-6}, 448]$. The range $S < 2^{-6}$ (entirely subnormal) is excluded for cleanliness; the conclusion still holds there since the uniform subnormal LSB $2^{-9}$ gives $dp(S) = 2^{-9}/S > 2^{-3}$, but the practical scope of this paper is $S \ge 1$.}

\textbf{(iii)} For $S > 448$, both terms in $\max(32/S, 2(1-448/S))$ exceed $2^{-4}$ over disjoint sub-ranges: $32/S > 2^{-4}$ for $S < 512$, and $2(1 - 448/S) > 2^{-4}$ for $S > 448 \cdot \frac{1}{1 - 2^{-5}} \approx 462.5$; the two intervals overlap on $(448, 512)$, covering the full $S > 448$ regime.
\end{proof}

Figure~\ref{fig:dp} shows $dp(S)$ computed by exact enumeration of all 126 positive E4M3 values, confirming the sawtooth structure and power-of-two lower envelope.

\begin{figure}[t]
\centering
\includegraphics[width=\columnwidth]{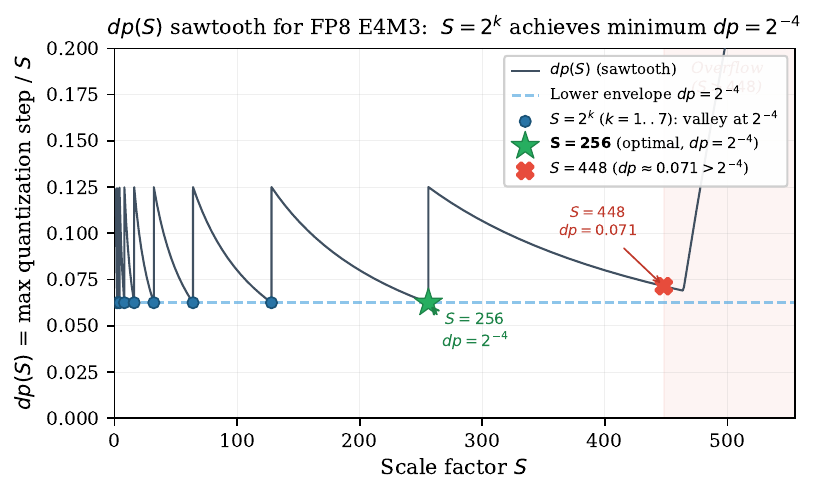}
\caption{The normalized quantization step $dp(S)$ for E4M3.
All $2^k$ values (blue) sit on the lower envelope at $dp = 2^{-4}$.
$S = 256$ (green star) is the rightmost before overflow.
$S = 448$ (red) sits $14\%$ above the envelope.}
\label{fig:dp}
\end{figure}

\subsection{Condition 3: Maximum Normal Coverage}

Among the power-of-two candidates satisfying $S \leq 448$ (i.e., $S \in \{1, 2, 4, \ldots, 256\}$), larger $S$ is better because the E4M3 \emph{normal-region lower bound} in the P-domain scales as $2^{-6}/S$:

\begin{center}
\small
\begin{tabular}{@{}lcc@{}}
\toprule
Scale & P normal threshold & Coverage \\
\midrule
$S = 64$ & $2^{-6}/64 = 2^{-12} \approx 2.4 \times 10^{-4}$ & Baseline \\
$S = 128$ & $2^{-6}/128 = 2^{-13} \approx 1.2 \times 10^{-4}$ & $2\times$ better \\
\rowcolor{lightgreen}
$S = 256$ & $2^{-6}/256 = 2^{-14} \approx 6.1 \times 10^{-5}$ & $4\times$ better \\
$S = 512$ & \multicolumn{2}{c}{Overflow (exceeds 448)} \\
\bottomrule
\end{tabular}
\end{center}

\noindent This monotonicity does \emph{not} stop at $256$: a non-power-of-two scale such as $S = 448$ pushes the normal threshold down further to $2^{-6}/448 \approx 3.5 \times 10^{-5}$, i.e.\ \emph{strictly better} coverage than $S = 256$ ($6.1 \times 10^{-5}$). Coverage alone therefore favors $448$; $256$ wins only once we additionally require bit-exactness (C1) and the minimum quantization step (C2), both of which $448$ fails. The $256$-vs-$448$ choice is thus a genuine trade-off (smaller worst-case step vs.\ slightly better deep-tail coverage), resolved empirically in \S\ref{sec:experiments}---not a clean domination.

\subsection{Characterization of $S = 256$}

\begin{theorem}[Characterization of $S = 256$]
\label{thm:256}
Restricted to scales $S \geq 1$ (so that the cast \emph{expands} the dynamic range of $P \in [0,1]$ rather than contracting it), $S = 256 = 2^8$ is the unique value satisfying all three of:
\begin{enumerate}[label=\textup{(C\arabic*)},leftmargin=*]
\item $S = 2^k$ \quad (bit-exact scaling; Proposition~\ref{prop:bitexact})
\item $dp(S) = 2^{-4}$ \quad (minimum quantization step; Theorem~\ref{thm:dp})
\item $S = \max\{2^k : 2^k \leq 448\}$ \quad (maximum normal coverage \emph{among $2^k$ scales})
\end{enumerate}
\end{theorem}

\begin{proof}
By (C1), $S = 2^k$ for some integer $k \geq 0$ (using the standing $S \geq 1$ assumption).
By Theorem~\ref{thm:dp}, the set $\{2^k : k \geq 0\}$ partitions into $\{2^0,\ldots,2^8\}$ where $dp = 2^{-4}$ and $\{2^k : k \geq 9\}$ where $dp > 2^{-4}$ (Theorem~\ref{thm:dp}(iii)).
Thus (C2) admits exactly the nine candidates $\{1, 2, 4, \ldots, 256\}$.
(C3) selects the maximum element of this set, $k = 8$, giving $S = 256$.
\end{proof}

\begin{remark}[Discriminating power of the conditions]
Of the three conditions, (C3) alone---``pick the largest $2^k \leq 448$''---uniquely identifies $S = 256$ among power-of-two scales.
(C1) and (C2) jointly leave nine candidates and merely articulate \emph{why} non-$2^k$ scales and $2^k$ with $k \geq 9$ are dominated.
The value of (C2) is therefore explanatory rather than discriminating: it makes the lower-envelope structure of the $dp$ sawtooth (Figure~\ref{fig:dp}) explicit, and shows that the bit-exactness preference (C1) does not come at the cost of quantization-step optimality.
Finally, (C3) is a \emph{within-$2^k$} statement: if bit-exactness is dropped, $S = 448$ achieves strictly larger normal coverage (\S\ref{sec:scale}). The optimality claimed here is thus for the bit-exact family, and the practical $256$-vs-$448$ gap is the empirical $dp$ (worst-case step) effect, not a coverage advantage.
\end{remark}

\paragraph{Comparison: $S = 256$ vs $S = 448$.}
$dp(448) = 32/448 \approx 0.0714$ exceeds $dp(256) = 0.0625$ by $14\%$.
In MSE terms: $(0.0714/0.0625)^2 \approx 1.30$, predicting $\sim$30\% higher MSE for $S = 448$.
Our experiments (\S\ref{sec:experiments}) measure 10--15\% MSE difference, consistent with the bound (the bound is worst-case; average-case is milder, and $448$'s better deep-tail coverage partly offsets its larger worst-case step).

\section{Experimental Validation}
\label{sec:experiments}

\subsection{Kernel-Faithful Simulation}
\label{sec:sim_model}

We simulate the FP8 attention kernel with semantics matched to production code (specifically, the \texttt{attention\_with\_kvcache\_prefill\_fp8} kernel in hpc-ops and FlashAttention-3's Hopper/Blackwell backend):

\begin{itemize}[leftmargin=*,itemsep=1pt]
\item $\ell$ (running sum): FP32, from \emph{pre-cast} P (line 8 of Alg.~\ref{alg:fp8attn}).
\item $O$ (output): FP32 accumulator, from \emph{post-cast} $P_{\text{fp8}} \times V$.
\item P-cast: round-to-nearest E4M3 (values below $2^{-10}$ $\to$ 0).
\item Epilogue: $O_{\text{final}} = O / (S \cdot \ell)$.
\end{itemize}

To \emph{isolate} the P-cast effect, we keep $Q, K, V$ in FP32.
The QKV FP8 quantization is an orthogonal concern (handled by separate qscale/kscale/vscale) and does not interact with the P-cast precision loss.

\subsection{Configurations}

We compare five configurations spanning the design space:

\begin{center}
\small
\begin{tabular}{@{}llc@{}}
\toprule
Label & Parameters & Matches \\
\midrule
Forward, $S{=}1$ & fwd, no scale & hpc-ops (prior) \\
Reverse, $S{=}1$ & rev, no scale & --- \\
Forward, $S{=}448$ & fwd, max\_normal & FlashInfer/TRT \\
Forward, $S{=}256$ & fwd, $2^8$ & \textbf{This work} \\
Reverse, $S{=}256$ & rev, $2^8$ & FA3/4 \\
\bottomrule
\end{tabular}
\end{center}

\subsection{Experiment 1: Sink Strength Sweep}

We sweep $\Delta \in [4, 13]$ with $N = 4096$, $d = 128$, $q_{\text{len}} = 32$, $B = 64$, $k_{\text{sink}} = 4$, over 20 seeds.
Figure~\ref{fig:delta} shows both the MSE results and the underlying physics (P-zeroing fraction and mass at risk).

\begin{figure}[t]
\centering
\includegraphics[width=\columnwidth]{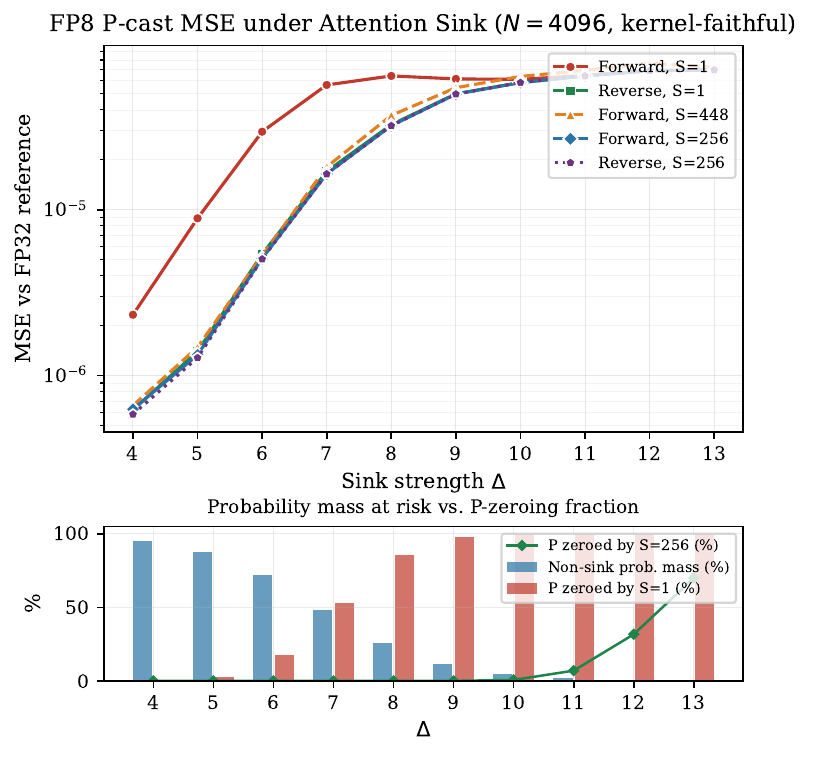}
\caption{\textbf{Top:} MSE vs.\ sink strength. Forward+S=1 exhibits a peak at $\Delta=7$--$8$ (the transition region). All other configurations remain at the quantization-noise floor.
\textbf{Bottom:} Diagnostic---non-sink probability mass (blue bars) and fraction of P values zeroed by S=1 cast (red bars) vs.\ S=256 cast (green line).}
\label{fig:delta}
\end{figure}

\paragraph{Key findings:}
\begin{enumerate}[leftmargin=*,itemsep=1pt]
\item Forward+S=1 is $3.4\times$ worse than optimized configurations at $\Delta = 7$.
\item Both reverse (any $S$) and $S = 256$ (any direction) independently fix the issue.
\item $S = 256$ gives $\sim$10--15\% lower MSE than $S = 448$, consistent with the $dp$ ratio.
\item At $\Delta \geq 11$, all configurations converge (non-sink mass $< 2\%$).
\end{enumerate}

\subsection{Experiment 2: Sequence Length Sweep}

At $\Delta = 7$, we sweep $N \in [512, 16384]$.
As $N$ grows, the non-sink probability mass increases, amplifying the P-collapse effect.
Figure~\ref{fig:seqlen} shows the results.

\begin{figure}[t]
\centering
\includegraphics[width=\columnwidth]{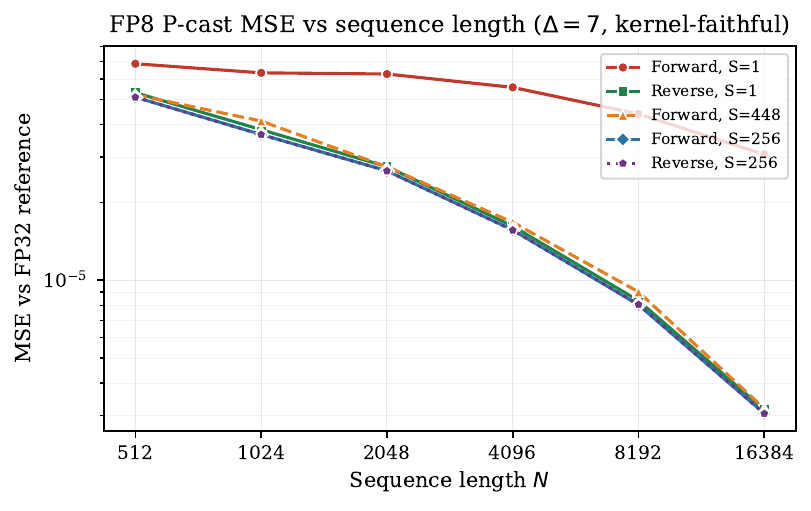}
\caption{MSE vs.\ sequence length at $\Delta = 7$.
The improvement of Forward+S=256 over Forward+S=1 grows from $1.3\times$ ($N{=}512$) to $10\times$ ($N{=}16384$) as non-sink probability mass increases.}
\label{fig:seqlen}
\end{figure}

\subsection{Numerical Results}

Table~\ref{tab:results} reports exact MSE values for key operating points:

\begin{table}[h]
\centering
\caption{MSE ($\times 10^{-5}$) at selected configurations.}
\label{tab:results}
\small
\begin{tabular}{@{}l|ccc@{}}
\toprule
Configuration & $\Delta{=}7$ & $\Delta{=}7$ & $\Delta{=}7$ \\
 & $N{=}4096$ & $N{=}8192$ & $N{=}16384$ \\
\midrule
Forward, $S{=}1$ & 5.65 & 4.40 & 2.94 \\
Reverse, $S{=}1$ & 1.70 & 0.83 & 0.32 \\
Forward, $S{=}448$ & 1.81 & 0.90 & 0.32 \\
\rowcolor{lightgreen}
Forward, $S{=}256$ & 1.64 & 0.80 & 0.28 \\
Reverse, $S{=}256$ & 1.64 & 0.81 & 0.28 \\
\midrule
Ratio (Fwd$_{S=1}$ / best) & $3.4\times$ & $5.5\times$ & $10.5\times$ \\
\bottomrule
\end{tabular}
\end{table}

\subsection{Comparison with Mainstream Implementations}

Table~\ref{tab:impls} positions the implementations in terms of both design choices.

\begin{table}[h]
\centering
\caption{Design choices in production FP8 attention kernels.
``Optimal'' denotes joint satisfaction of (C1)--(C3) of Theorem~\ref{thm:256} \emph{and} reverse iteration (which is redundant given $S{=}256$ but offers belt-and-suspenders robustness against multi-sink patterns).
``Near-optimal'' denotes one fix applied (here, reverse iteration eliminates P-collapse) but a non-power-of-two scale incurs the $dp(448)/dp(256)\!=\!1.14$ residual penalty (\S\ref{sec:scale}).
``Suboptimal scale'' denotes the same scale issue \emph{without} the iteration-order safety net.}
\label{tab:impls}
\small
\begin{tabular}{@{}lcccc@{}}
\toprule
Implementation & Order & $S$ & $2^k$? & Status \\
\midrule
FlashAttention-3/4 & Rev & 256 & \checkmark & Optimal \\
hpc-ops (updated) & Fwd & 256 & \checkmark & Optimal \\
hpc-ops (prior) & Fwd & 1 & \checkmark & $S{=}1$ baseline \\
FlashInfer & Rev & 448 & $\times$ & Near-optimal \\
TensorRT-LLM XQA & Fwd & 448 & $\times$ & Suboptimal scale \\
SageAttention2 & Fwd & 448 & $\times$ & Suboptimal scale \\
\bottomrule
\end{tabular}
\end{table}

\noindent Based on this analysis, we have adopted $S = 256$ in the hpc-ops codebase.
FlashAttention-3/4 and the updated hpc-ops both satisfy the optimality conditions; FlashInfer and TensorRT-LLM XQA could gain an additional 10--15\% precision improvement by switching from $S = 448$ to $S = 256$.

\section{Discussion}
\label{sec:discussion}

\subsection{Saturation: Same Mechanism, Same Floor}

Both optimizations address the \emph{identical} failure mode: P values falling below E4M3's representable range.
Once either fix is applied, P-collapse is eliminated and the residual MSE is set by the \emph{inherent} E4M3 quantization noise on representable values.
Paired $t$-tests over 100 instances confirm that Forward+S=256 and Reverse+S=256 are statistically indistinguishable wherever P-collapse is active ($\Delta \le 9$); for $\Delta \ge 10$ the residuals diverge with reverse marginally better, but the absolute gap is $\sim 10^{-8}$, three orders of magnitude below the MSE itself, so the practical conclusion is unchanged (Appendix~\ref{app:saturation}).

This has a practical implication: \emph{either optimization alone suffices}.
The choice between them should be driven by engineering constraints (reverse may simplify causal mask handling; forward may offer better memory access patterns).

\subsection{Structural Interpretation}

The role of $S = 2^k$ admits a simple geometric reading: every E4M3 normal binade $[2^n, 2^{n+1})$ contains 8 equispaced points, differing only in absolute scale.
The quantity $dp(S) = \LSB(S)/S$ measures the \emph{normalized} granularity; at $S = 2^k$, both numerator and denominator double simultaneously, locking their ratio at $2^{-4}$.
Non-power-of-two scales break this alignment, leaving $dp(S)$ on the rising portion of the sawtooth (Figure~\ref{fig:dp}).
The P-collapse threshold $\Delta_c = 6.93 + \ln S - \delta_k$ corresponds to the sink strength at which the median non-sink $P$ value falls below the round-to-zero boundary $2^{-10}/S$; below this threshold the cast is well-conditioned, above it most of the non-sink mass is silently zeroed.

\subsection{Practical Recommendations}

\begin{enumerate}[leftmargin=*,itemsep=2pt]
\item \textbf{For forward-order kernels (e.g., TRT-LLM XQA):} Add $S = 256$ P-scaling.
Implementation: multiply P by 256 before cast; divide output by 256 in epilogue.
Both operations are bit-exact (\S\ref{sec:scale}) and add only two FMAs whose latency is negligible against the tensor-core matmul.
We have applied this change to the hpc-ops kernel.
\item \textbf{For kernels already using $S = 448$ (FlashInfer, SageAttention2):} Switch to $S = 256$ for 10--15\% MSE reduction at negligible cost.
\item \textbf{For kernels already using reverse + $S = 256$ (FA3/4):} No change needed---this is already optimal among static-scale strategies.
\end{enumerate}

\subsection{Limitations and Future Work}

\textbf{Worst-case vs.\ average-case.}
The $dp(S)$ analysis provides the minimax-optimal scale assuming $P$ uniformly spans $[0, 1]$.
In practice, the post-softmax distribution is highly non-uniform (a few large values, many small).
A \emph{dynamic} per-block scale (still constrained to $2^k$ for bit-exactness) could yield better average-case precision by adapting to the per-block $P$ distribution.
Even so, for typical small $P$ ($\sim$0.01--0.05) the static $S = 256$ maps into mid-range binades with competitive relative precision, and its no-underflow guarantee remains the main advantage.

\textbf{Interaction with QKV quantization.}
We isolated the P-cast effect by keeping $Q, K, V$ in FP32.
In production, these are also in FP8, contributing an additive noise floor.
Since both reference and kernel use identical dequantized QKV values, the P-cast MSE signal is preserved, but the \emph{relative} improvement ratio is compressed at high delta where QKV noise dominates.

\textbf{No end-to-end metric.}
All our numbers are isolated output MSE against an FP32 reference; we report no perplexity or task-accuracy. The $3$--$10\times$ figures show the P-cast error is real and removable, but \emph{not} that removing it shifts a downstream metric---where the FP8 QKV floor dominates, the end-to-end benefit may be small. We thus frame the $S{=}256$ switch as zero-cost and strictly-no-worse rather than a guaranteed quality win, to be confirmed per model.

\textbf{Scope of the $dp(S)$ analysis.}
The $dp(S)$ construction assumes (a)~a bounded range (here $P \in [0,1]$), (b)~one static scale per cast, and (c)~a bit-exactness requirement. These break down elsewhere: MXFP4's E8M0 block exponent already forces scales to $2^k$ (removing the freedom $dp(S)$ addresses), and outlier-heavy activations need per-channel smoothing first. We thus make no claims beyond E4M3 P-cast.

\textbf{Multi-sink patterns.}
Our analysis assumes a single sink cluster at position 0.
Some models exhibit distributed sink patterns across the sequence; the P-collapse analysis generalizes straightforwardly by replacing $\Delta$ with the per-block maximum gap.

\section{Conclusion}
\label{sec:conclusion}
\begingroup\samepage

We have analyzed two implementation-level precision considerations for FP8 E4M3 attention:
(1)~P-collapse under Attention Sink is a quantifiable threshold effect---activating around $\Delta \approx 6$--$7$ with leading-order underflow fraction $F = \Phi(\Delta + \delta_k - 6.93 - \ln S)$, where $\delta_k$ is the within-sink extreme-value shift---that both reverse iteration and $S = 256$ scaling independently eliminate;
(2)~the $dp(S)$ sawtooth, defined over the E4M3 number line, characterizes $S = 256$ via three conditions---bit-exact arithmetic, the lower-envelope minimum step ($dp = 2^{-4}$), and the largest normal-range coverage among $2^k$ scales.
Kernel-faithful experiments measure $3$--$10\times$ MSE improvement at the critical transition region ($\Delta = 5$--$9$, $N = 4096$--$16384$), and paired statistical tests show that both optimizations reach the same precision floor when combined.

For practitioners, the actionable recommendation is simple: for any FP8 attention kernel currently using $S = 1$ (direct cast) or $S = 448$ (max-normal), switching to $S = 256$ is a single-constant, bit-exact change that removes P-collapse and is never worse on P-cast MSE; whether it moves an end-to-end metric should be confirmed per model (\S\ref{sec:discussion}).
We make no transfer claims to other formats (MXFP4, FP4) or per-channel quantization, where the preconditions differ (\S\ref{sec:discussion}).
\endgroup

\vspace{-2mm}
{\small
\bibliographystyle{plainnat}
\bibliography{references}
}

\appendix

\section{Complete $dp(S)$ Verification}
\label{app:dp}

We verify Theorem~\ref{thm:dp} by exhaustive computation over all 126 positive E4M3 values: for each $S$, take the binade with maximum LSB overlapping $[0,S]$ and compute $dp(S) = \text{LSB}_{\max}/S$.

\vspace{1mm}
\noindent\textbf{Power-of-two scales} ($k = 0, 1, \ldots, 8$): all give $dp = 2^{-4} = 0.0625$ exactly. Top binade is $[2^{k-1}, 2^k)$ with LSB $= 2^{k-4}$. For $k = 9$ ($S = 512$): overflow, $dp = 0.25$.

\vspace{1mm}
\noindent\textbf{Non-power-of-two scales}: $dp(3) = 0.083$, $dp(100) = 0.080$, $dp(250) = 0.064$, $dp(300) = 0.107$, $dp(448) = 0.071$. All strictly exceed $2^{-4}$.
See \texttt{experiments/proof\_mathematical.py} for full enumeration.

\section{Saturation Statistical Verification}
\label{app:saturation}

We test $H_0$: MSE(Fwd+$S{=}256$) $=$ MSE(Rev+$S{=}256$) via paired $t$-test (100 instances per condition).

\begin{center}
\footnotesize
\begin{tabular}{@{}ccrcc@{}}
\toprule
$\Delta$ & $N$ & Mean diff & $t$ & Sig? \\
\midrule
7 & 4k & $-7{\times}10^{-9}$ & $-1.0$ & No \\
7 & 16k & $-3{\times}10^{-9}$ & $-1.1$ & No \\
10 & 16k & $-2{\times}10^{-8}$ & $-4.4$ & Yes$^*$ \\
12 & 16k & $-1{\times}10^{-8}$ & $-3.6$ & Yes$^*$ \\
\bottomrule
\end{tabular}
\end{center}
\noindent${}^*$Where significant, reverse is marginally \emph{better} (not worse). All differences are $\sim\!10^{-8}$, negligible vs.\ MSE values of $10^{-5}$--$10^{-6}$, confirming algorithmic equivalence.

\end{document}